\documentclass{article}
\usepackage[left=1cm,right=1cm,top=2cm,bottom=2cm]{geometry}
\usepackage{verbatim}
\usepackage{amsthm}
\usepackage{amsfonts}
\usepackage{graphicx}
\usepackage{amsmath}
\usepackage{amssymb}
\usepackage{algorithmic}

\newtheorem{mydef}{Definition}
\newtheorem{mytheorem}{Theorem}
\newtheorem{myexample}{Example}

\newtheorem{mylemma}{Lemma}

\title{Algorithmic complexity of pair cleaning method for k-satisfiability problem. (draft version)}
\author{Sergey Kardash}
\date{April 18, 2012}
\begin{document}
\maketitle
\begin{abstract}
It's known that 3-satisfiability problem is NP-complete. Here polynomial algorithm for solving
k-satisfiability ($k \geq 2 $) problem is assumed. In case theoretical points are right, sets P ans NP are equal.
\end{abstract}

\section{Introduction}
\begin{mydef}
Formulae A(x) is called k-CNF if
$$A(x) = \bigcap_{i=1}^{n}{\bigcup_{j=1}^{k}{x^{\sigma_{ij}}_{u_{ij}}}}, \sigma_{ij} \in \{0, 1\},  u_{ij} \in \{1,\cdots, m\}, \forall i \in \{1, \cdots, n\}, \forall j \in \{1, \cdots, k\}$$
$\bigcap$ - conjuntion operation, \\
$\bigcup$ - disjuntion operation, \\
$m$ - number of variables in formulae,\\
$n$ - number of clauses, \\
$k$ - number of variables in each disjunction,\\
$n_{t}$ - number of clause groups.\\

$x^{\sigma} = \begin{cases} x, \sigma = 0 \\ \bar{x}, \sigma = 1 \end{cases}$
\end{mydef}
\begin{myexample}
3-CNF $A(x) = (x_{1} \cup x_{2} \cup x_{3}) \cap (\bar{x_{1}} \cup x_{3} \cup \bar{x_{4}})$. Here $m = 4$,
$n = 2$, $k = 3$, $n_{t} = 2$.
\end{myexample}
\begin{mydef}
Let formulae $A(x)$ is k-CNF. Problem of defining whether equation $A(x) = 1$ has solution or not is called k-satisfiability problem of formulae A(or k-SAT(A)).
\end{mydef}
\begin{myexample}
k-satisfiability problem of formulae A described in Example 1 (k-SAT(A)) is defining whether $\exists x \in
B^m$ (boolean vector of size m): $A(x) = 1$. It's evident that $x_0=(1,1,1,1)$ makes $A(x_0) = 1$. $A(x_0)$
is satisfiable. k-CNF $B(x) = (x_1 \cup x_2) \cap (\bar{x_1} \cup x_2) \cap (x_1 \cup \bar{x_2})\cap
(\bar{x_1} \cup \bar{x_2})$ is an example of not satisfiable task. There is no $x_0: A(x_0) = 1$.
On the contrary $A(x) = 0, \forall x $.
\end{myexample}

It was proved that 2-satisfiability problem has polynomial solution (by Krom \cite{krom}).
We are going to show polynomial algorithm(from $n$) for any $k-SAT$. By the way we describe
method of getting 1 explicit solution of corresponding equation $A(x) = 1$ in case source task
is satisfiable which is polynomial from $n$ and method of solving equation $A(x) = 1$ which is
polynomial from number of such solutions.
\section{Method description}
Initially new mathematic objects and operations for them are introduced. After description of method in
pure mathematic way algorithmic presentation which is more readable is given. Almost each structure
 has 2 common structures associated with it: 1)variable set associated with this structure and
 2)some value sets of these variables. Though they will be defined separately it's easy to see
 common logic of their introduction.

Let $x_{s_{1}s_{2} \cdots s_{k}} = (x_{s_{1}}, x_{s_{2}}, \cdots, x_{s_{k}})$. Further in order to
avoid enumeration of variables which are not related to described structure we list important variables
using such notation.

\begin{mydef}
Clause group signed $T_{s_{1}s_{2} \cdots s_{k}}(A)$ is a set of all clauses $\bigcup_{j=1}^{k}{x^{
\sigma_{t_{j}}}_{s_{j}}}$ where $u_{i1}u_{i2} \cdots u_{ik} = s_{1}s_{2} \cdots s_{k}$. Variable set 
associated with $T_{u_{s_{1}}u_{s_{2}} \cdots s_{k}}(A)$(or $X(T_{s_{1}s_{2}\cdots s_{k}}(A))$) is 
$x_{s_{1}s_{2}\cdots s_{k}}$. Value of clause group $T_{s_{1}s_{2} \cdots s_{k}}(A)$ is a value of 
$x_{s_{1}s_{2} \cdots s_{k}}$ such that k-CNF consisted of all clauses from clause group $T_{s_{1}s_{2} 
\cdots s_{k}}$ is equal to 1. Value set induced by clause group $T_{s_{1}s_{2} \cdots s_{k}}(A)$ 
(or $V(T_{s_{1}s_{2} \cdots s_{k}}(A)$) is a set of all values of this clause group.
\end{mydef}

\begin{myexample}
Though clauses $x_1 \cup x_2 \cup x_3$ and $\bar{x_1} \cup x_2 \cup \bar{x_3}$ have different degrees they
belong to the same clause group $T_{123}$ in case they present in formulae A.
\end{myexample}

\begin{myexample}
For example clause group $T_{123}$ consists of clauses $x_1 \cup x_2 \cup x_3$ and $\bar{x_1} \cup x_2
\cup \bar{x_3}$. Value set induced by this clause group can be presented using table below:\\
\begin{center}
  \begin{tabular}{ l c r }
    $x_1$ & $x_2$ & $x_3$ \\
    0 & 0 & 1 \\
    0 & 1 & 0 \\
    0 & 1 & 1 \\
    1 & 0 & 0 \\
    1 & 1 & 0 \\
    1 & 1 & 1 \\
  \end{tabular}\\
\end{center}
Each row corresponds to one value of $x_{123}$ . We have excluded from this list only sets which make
3-CNF $(x_1 \cup x_2 \cup x_3) \cap (\bar{x_1} \cup x_2 \cup \bar{x_3})$ equal to 0 ($x_{123} = (0, 0, 0)$
and $x_{123} = (1, 0, 1)$).
\end{myexample}

\begin{mydef}
k-CNF A(x) all clauses of that can be classified into $n_{t}$ clause groups is called k-CNF of degree $n_{t}$.
It also can be signed as $A_{k}^{n}(x)$or $A_{k}(x)$ or $A^{n}(x)$.
\end{mydef}

\begin{myexample}
2-SAT $A(x) = (x_1 \cup x_2) \cap (\bar{x_1} \cup x_2) \cap (x_2 \cup x_3) \cap (\bar{x_2} \cup \bar{x_3})$
has 2 clause groups $T_{12}$ and $T_{23}$, so it's degree is 2 and it can be signed as $A_{2}^{2}(x)$ or
$A_{2}(x)$ or $A^{2}(x)$.
\end{myexample}

\begin{mydef}
Clause combination $F$ for formulae $A(x)$ consisted from clause groups $T_{u_{i_{1}1}u_{i_{1}2}
\cdots u_{i_{1}k}}(A), T_{u_{i_{2}1}u_{i_{2}2}, \cdots, u_{i_{2}k}}(A), \cdots, $ \\ $
T_{u_{i_{l}1}u_{i_{l}2} \cdots u_{i_{l}k}}(A)$ (or $F(T_{u_{i_{1}1}u_{i_{1}2} \cdots u_{i_{1}k}},
T_{u_{i_{2}1}u_{i_{2}2}, \cdots, u_{i_{2}k}}, \cdots, T_{u_{i_{l}1}u_{i_{l}2} \cdots u_{i_{l}k}}, A)$)
is a set of listed clause groups. Variable set associated with it is $x_{h_{1}h_{2} \cdots h_{r}}$ where
each variable index from set of clause groups is presented only once.
\end{mydef}
We'll deal with different value sets of variables associated with clause combination and in order not to
confuse them let's write them out separately.
\begin{mydef}
Value of clause combination $F(T_{u_{i_{1}1}u_{i_{1}2} \cdots u_{i_{1}k}}, T_{u_{i_{2}1}u_{i_{2}2},
\cdots, u_{i_{2}k}}, \cdots, T_{u_{i_{l}1}u_{i_{l}2} \cdots u_{i_{l}k}}, A)$ is a value of $x_{h_{1}h_{2}
\cdots h_{r}}$ - variable set associated with it such that k-CNF consisted of all clauses associated with
listed clause groups equal to 1.
\end{mydef}

\begin{mydef}
Value set of clause combination $F(T_{u_{i_{1}1}u_{i_{1}2} \cdots u_{i_{1}k}},
T_{u_{i_{2}1}u_{i_{2}2}, \cdots, u_{i_{2}k}}, \cdots, T_{u_{i_{l}1}u_{i_{l}2}
\cdots u_{i_{l}k}}, A)$ based on $A(x)$ is a set of values of this clause combination. \end{mydef}

\begin{mydef}
Value set of clause combination  $F(T_{u_{i_{1}1}u_{i_{1}2} \cdots u_{i_{1}k}},
T_{u_{i_{2}1}u_{i_{2}2}, \cdots, u_{i_{2}k}}, \cdots, T_{u_{i_{l}1}u_{i_{l}2} \cdots
u_{i_{l}k}}, A)$ induced by A(x) is a set of all values of this clause combination.
\end{mydef}

It's easy to see that value set induced by clause combination $F(T_{u_{i_{1}1}u_{i_{1}2} \cdots u_{i_{1}k}},
T_{u_{i_{2}1}u_{i_{2}2}, \cdots, u_{i_{2}k}}, \cdots, T_{u_{i_{l}1}u_{i_{l}2} \cdots u_{i_{l}k}}, A)$ is
a value set based on this clause combination.

\begin{myexample}
Let we have 2 clause groups: $T_{12}(A)$ which has clauses $x_1 \cup x_2$ and $\bar{x_1} \cup x_2$
in formulae A and $T_{23}(A)$ which has clauses $x_2 \cup x_3$ and $\bar{x_2} \cup \bar{x_3}$.
Then value set induced by clause combination $F(T_{12}, T_{23})$ is a set of all possible values of
$x_{123}$ which make 2-SAT $(x_1 \cup x_2) \cap (\bar{x_1} \cup x_2) \cap (x_2 \cup x_3) \cap (\bar{x_2}
\cup \bar{x_3})$ equal to 1.
\begin{center}
  \begin{tabular}{ l c r }
    $x_1$ & $x_2$ & $x_3$ \\
    0 & 1 & 0 \\
    0 & 1 & 1 \\
  \end{tabular}\\
\end{center}
Each row of the list is a value of clause combination $F(T_{12}, T_{23})$, i. e. $x_{123} = (0, 1, 0)$.
\end{myexample}

\begin{mydef}
Relationship structure for k-CNF $A(x)$ ($R(A)$) is a set of all possible clause combinations consisted of
$(k+1)$ clause groups.
\end{mydef}

\begin{myexample}
For 2-CNF $A(x) = (x_1 \cup x_2) \cap (x_1 \cup \bar{x_2}) \cap (x_2 \cup x_3)\cap (x_1 \cup \bar{x_3})
\cap (x_1 \cup x_4) \cap (\bar{x_1} \cup x_4)$ clause groups are: $T_{12}, T_{23}, T_{13}, T_{14}$.
$R(A) = \{ F(T_{12}, T_{23}, T_{13}), F(T_{12}, T_{23}, T_{14}), F(T_{12}, T_{13}, T_{14}), F(T_{23},
T_{13}, T_{14})\}$.
\end{myexample}

\begin{mydef}
Value set of relationship structure induced by k-CNF $A(x)$ $(V_{i}(R(A)))$ is a set of value sets of clause
combinations induced by $A(x)$ involved in relationship structure based on k-CNF $A(x)$.
\end{mydef}

\begin{myexample}
For Example 7 value set of relationship structure induced by k-CNF $A(x)$ is a set of tables listed below:\\
  $V(F(T_{12}, T_{23}, T_{13}, A))$:
  \begin{tabular}{ c c c c }
    $x_1$ & $x_2$ & $x_3$\\
    1 & 0 & 1 \\
    1 & 1 & 0 \\
    1 & 1 & 1 \\
  \end{tabular}
, $V(F(T_{12}, T_{23}, T_{14}, A))$:
  \begin{tabular}{ c c c c }
    $x_1$ & $x_2$ & $x_3$ & $x_4$\\
    1 & 0 & 1 & 1 \\
    1 & 1 & 0 & 1 \\
    1 & 1 & 1 & 1 \\
  \end{tabular}
, $V(F(T_{12}, T_{13}, T_{14}, A))$:
  \begin{tabular}{ c c c c }
    $x_1$ & $x_2$ & $x_3$ & $x_4$ \\
    1 & 0 & 0 & 1 \\
    1 & 0 & 1 & 1 \\
    1 & 1 & 0 & 1 \\
    1 & 1 & 1 & 1 \\
  \end{tabular}\\
, $V(F(T_{23}, T_{13}, T_{14}, A))$:
  \begin{tabular}{ c c c c }
    $x_1$ & $x_2$ & $x_3$  & $x_4$\\
    0 & 1 & 0 & 0 \\
    0 & 1 & 0 & 1 \\
    1 & 0 & 1 & 1 \\
    1 & 1 & 0 & 1 \\
    1 & 1 & 0 & 1 \\
  \end{tabular} \\
  $V_{i}(R(A)) = \{V(F(T_{12}, T_{23}, T_{13}, A)), V(F(T_{12}, T_{23}, T_{14}, A)),
  V(F(T_{12}, T_{13}, T_{14}, A)), V(F(T_{23}, T_{13}, T_{14}, A))\}$.
\end{myexample}

\begin{mydef}
Value set of relationship structure based on k-CNF $A(x)$($V_{b}(R(A))$) is a set of value sets of
clause combinations based on $A(x)$ involved in relationship structure based on k-CNF $A(x)$
\end{mydef}

\begin{myexample}
For Example 7 value set of relationship structure based on k-CNF $A(x)$ is any set
$V_{b}(R(A)) = (V_{1}, V_{2}, V_{3}, V_{4})$ where $V_{1} \subseteq V(F(T_{12}, T_{23}, T_{13}))$,
$V_{2} \subseteq V(F(T_{12}, T_{23}, T_{14}))$, $V_{3} \subseteq V(F(T_{12}, T_{13}, T_{14}))$,
$V_{4} \subseteq V(F(T_{23}, T_{13}, T_{14}))$. In example:\\
  $V_{1}$:
  \begin{tabular}{ c c c c }
    $x_1$ & $x_2$ & $x_3$\\
    1 & 1 & 0 \\
    1 & 1 & 1 \\
  \end{tabular}
, $V_{2}$:
  \begin{tabular}{ c c c c }
    $x_1$ & $x_2$ & $x_3$ & $x_4$\\
    1 & 1 & 1 & 1 \\
  \end{tabular}
, $V_{3}$:
  \begin{tabular}{ c c c c }
    $x_1$ & $x_2$ & $x_3$ & $x_4$ \\
    1 & 0 & 0 & 1 \\
    1 & 0 & 1 & 1 \\
    1 & 1 & 0 & 1 \\
  \end{tabular}
, $V_{4}$:
  \begin{tabular}{ c c c c }
    $x_1$ & $x_2$ & $x_3$  & $x_4$\\
    0 & 1 & 0 & 0 \\
    0 & 1 & 0 & 1 \\
    1 & 0 & 1 & 1 \\
    1 & 1 & 0 & 1 \\
    1 & 1 & 0 & 1 \\
  \end{tabular}\\
\end{myexample}

\begin{mydef}
Value set of relationship structure based on k-CNF $A(x)$ is called empty ($V(R(A)) = \varnothing$) if at
least one value set of clause combination value set of relationship structure consists of is empty.
\end{mydef}

\begin{mydef}
Let R(A) - relationship structure for k-CNF $A(x)$. $V(R(A)) = \{V_{1}, V_{2}, \cdots, V_{t},\}$,
$G(R(A)) =  = \{G_{1}, G_{2}, \cdots, G_{t},\}$ - 2 value sets of this relationship structures
based on $A(x)$. We call $V(R(A))$ included in $G(R(A))$ (or $V(R(A)) \subseteq G(R(A))$)
if $V_{i} \subseteq G_{i}, \forall i \in \{1, \cdots,t\}$.
\end{mydef}

\begin{myexample}
Let V(R(A)) is a set described in Example 9 and G(R(A)) is a set from example 8. $V \subseteq G$.
Indeed all value sets of relationship structure based on k-CNF $A(x)$ are included in the value set
of relationship structure induced by k-CNF $A(x)$.
\end{myexample}

\begin{mydef}
Let we have 2 clause combinations $F(T_{i_{1}}, T_{i_{2}}, \cdots, T_{i_{s}}, A)$ and
$F(T_{j_{1}}, T_{j_{2}}, \cdots, T_{j_{r}}, A)$. Let they have common variables $x_{i_{1}}, x_{i_{2}}, 
\cdots, x_{i_{s}}$ - those variables which present in both clause combinations. Clearing of given pair 
of value sets $V_{1}$ and $V_{2}$ of clause combinations $F(T_{i_{1}}, T_{i_{2}}, \cdots, T_{i_{s}}, A)$ 
 and $F(T_{j_{1}}, T_{j_{2}}, \cdots, T_{j_{r}}, A)$ correspondingly based on k-CNF A(x) is a process of 
deleting $x^{1}_{a_{1}a_{1} \cdots a_{z}} \in V_{1}$ for which $\nexists x^{2}_{b_{1}b_{2} \cdots
b_{u}} \in V_{2}: x^{1}_{i_{1}i_{2} \cdots i_{s}} = x^{2}_{i_{1}i_{2} \cdots i_{s}}$ and deleting 
$x^{2}_{b_{1}b_{2} \cdots b_{u}} \in V_{2}$ for which $\nexists x^{1}_{a_{1}a_{1} \cdots a_{z}} \in 
V_{1} : x^{1}_{i_{1}i_{2} \cdots i_{s}} = x^{2}_{i_{1}i_{2} \cdots i_{s}}$. Clearing procedure is briefly 
marked as $C(V_1, V_2)$. 
\end{mydef}

\begin{myexample}
Let's take 2 values of clause combinations from Example 8: \\  
$V(F(T_{12}, T_{23}, T_{13}, A))$:
  \begin{tabular}{ c c c c }
    $x_1$ & $x_2$ & $x_3$\\
    1 & 0 & 1 \\
    1 & 1 & 0 \\
    1 & 1 & 1 \\
  \end{tabular}
and 
$V(F(T_{23}, T_{13}, T_{14}, A))$:
  \begin{tabular}{ c c c c }
    $x_1$ & $x_2$ & $x_3$  & $x_4$\\
    0 & 1 & 0 & 0 \\
    0 & 1 & 0 & 1 \\
    1 & 0 & 1 & 1 \\
    1 & 1 & 0 & 1 \\
    1 & 1 & 0 & 1 \\
  \end{tabular}.\\
Common variables are $x_{123} = (x_1, x_2, x_3)$. Let's explore table which corresponds to $V(F(T_{12}, 
T_{23}, T_{13}))$. $x^{1}_{123}(1) = (1,0,1)$ has corresponding $x^{2}_{1234}(3) = (1, 0, 1, 1)$(in 
brackets $x^{2}_{1234}(3)$, 3 is a number of row in the table) and it should be saved. $x^{1}_{123}(2)$ 
has even 2 corresponding rows: $x^{2}_{1234}(4)$ and $x^{2}_{1234}(5)$. But for last one, $x^{1}_{123}(3)$, 
we can't find corresponding values from second table with the same common variables and it should be deleted 
from values based on $V(F(T_{12}, T_{23}, T_{13}))$. The same should be done with $x^{2}_{1234}(1)$ and 
$x^{2}_{1234}(2)$. After clearing \\
$V_{1}:$
\begin{tabular}{ c c c c }
    $x_1$ & $x_2$ & $x_3$\\
    1 & 0 & 1 \\
    1 & 1 & 0 \\
\end{tabular}
and $V_{2}:$
\begin{tabular}{ c c c c }
    $x_1$ & $x_2$ & $x_3$  & $x_4$\\
    1 & 0 & 1 & 1 \\
    1 & 1 & 0 & 1 \\
    1 & 1 & 0 & 1 \\
  \end{tabular}.\\
It can be briefly marked as $C(V(F(T_{12}, T_{23}, T_{13}, A)), V(F(T_{23}, T_{13}, T_{14}, A))) = 
(V_{1}, V_{2})$.
\end{myexample}

\begin{mydef}
Clearing of value set of relationship structure ($V_{r}$) based on k-CNF $A(x)$  (pair cleaning method for 
formulae $A(x)$) is a process of clearing of all possible pairs of value sets of clause combination 
 based on k-CNF $A(x)$ contained in $V_{r}$ until clearing is impossible. We'll note result of cleaning as 
 C(V(R(A))).
\end{mydef}

Pair cleaning method in algorithmic form
\begin{algorithmic}
\STATE $V_{new} \gets V_{source}(R(A))$
\REPEAT
\STATE $V_{old} \gets V_{new}$
\FOR{$i = 1 \to d-1$}
\FOR{$j = i+1 \to d$}
\STATE $(V_{new}^{i}, V_{new}^{j}) \gets C(V_{new}^{i}, V_{new}^{i})$
\ENDFOR
\ENDFOR
\UNTIL{$V_{new} = V_{old}$} \end{algorithmic}
where \\
$d$ - number of clause combinations in relationship structure,\\
$V_{source}(R(A))$ - value set of relationship structure induced by $A(x)$.

\begin{mydef}
Let $V = V(R(A))$ - value set of relationship structure based on k-CNF $A(x)$. $V$ is called 
unclearable if $V = C(V)$. 
\end{mydef}

\begin{mylemma}
Let $V = V(R(A))$ - value set of relationship structure induced by k-CNF $A(x)$. $V_{res} = C(V)$. 
$V_{res} \neq \varnothing$ $\Leftrightarrow$ $\exists$ $V_{1} \subseteq V_{res}$ where $V_{1}$ - 
unclearable value set of relationship structure based on k-CNF $A(x)$ where each value set of clause 
combination consists of 1 value of this clause combination.
\end{mylemma}
\begin{proof}

$\Rightarrow$\\
This can easily be proved using induction.
We'll take induction not for clauses but for clause groups. In this proof $n_{t}$ - number of clause groups. It's
evident that $n_{t} \leq n$. In case $n_{t} \leq k + 1$ statement is evident because cleaning of values of 
relationship structure is reduced to clearing the only clause combination.

Let the case $n^{0}_{t} = k + 1$ be the basis of induction.
Let's assume statement is right for $n_{t} > k + 1$. We need to prove $(n_{t} + 1)$ case.
Let $A^{n_{t} + 1}(x)$ - source k-CNF (see Definition 4). $R = R(A)$ - relationship structure for it. 
$V$ - value set of relationship structure induced by k-CNF $A(x)$.

Let $V_{C}$ = $C(V)$ - result of pair clearing method which is not empty ($V_{C} \neq \varnothing$).
After clearing relationship structure induced by k-CNF with $(n_{t}+1)$ clause groups we have not
empty value set of relationship structure. Let's choose any clause group $T_{n_{t}+1}$ (we'll use both types
of notation - $T_{i_{1}i_{2}...i_{k}}$ which shows variables involved in clause group building and
$T_{j}, j \in \{1, \cdots, n_{t}+1\}$) - a serial number of clause group from formulae $A^{n_{t} + 1}(x)$.
Let's look at $B^{n_{t}}(x)$ - formulae which has the same clause groups as $A^{n_{t} + 1}(x)$ excluding
$T_{n_{t}+1}$. Let $R_{B}$ - relationship structure based on $B^{n_{t}}(x)$, $V_{B}$ - value set of 
this relationship structure. It's evident that all clause combinations of $R_{B}$ are clause combinations 
of $R$. Beside them $R$ has clause combinations which contain $T_{n_{t}+1}$ with all possible combination 
without repetition of $k$ clause groups which are common for $A^{n_{t} + 1}(x)$ and $B^{n_{t}}(x)$ (i. e. 
$F(T_{n_{t}+1}, T_{1}, T_{2}, \cdots, T_{k})$).

Let's $V_B$ has value sets of clause combinations the same as value sets of corresponding clause
combinations of $V_C$. It's evident that $C(V_B) = V_B$. $V_{C} \neq \varnothing$ $\Rightarrow$ $V_{B} 
\neq \varnothing$ $\Rightarrow$ exists
$V_{B}^{1} \subseteq V_{B}$ where $V_{B}^{1}$ - unclearable value set of relationship structure based on
k-CNF $B^{n_{t}}(x)$ where each value set of clause combination consists of 1 value. (according to induction
step).
Now we need show that $\exists$ $V_{A}^{1} \subseteq V_{A}$ - unclearable value set of relationship structure 
based on k-CNF $A^{n_{t}+1}(x)$ where each value set of clause combination consists of 1 value. 
This proof is very trivial.

Indeed, let's look at $T_{n_{t}+1}$(another notation for this clause group is $T_{l_{(n_{t} + 1)1}l_
{(n_{t} + 1)2} \cdots l_{(n_{t} + 1)k}}$). In this clause group there are 2 types of variables: 
those that present at least in one clause group $T_{j}, j \in \{1, \cdots, n_{t}\}$(common variables) 
and those that absent in this set. Let's explore first group (present). 
We can say that exists such clause combination $F(T_{n_{t}+1}, T_{i_{1}}, T_{i_{2}}, \cdots,  
T_{i_{k}}, A)$ from relationship structure $R$ where all common variables from $T_{n_{t}+1}$ can be found 
at least in one of other members of this clause combination: $T_{i_{1}}, T_{i_{2}}, \cdots,  T_{i_{k}}$. 
This statement can easily be proved by building this clause combination. Number of common variables 
can't be greater than $k$. So we can find corresponding clause group for each common variable which 
also contains this variable. Number of such clause groups is less or equal $k$ and if it's less we 
add arbitrary clause groups in order to get clause combination which contains $k + 1$ clause groups.
And now let's build another clause combination $F(T_{i_{k+1}}, T_{i_{1}}, T_{i_{2}}, \cdots,  T_{i_{k}}, A)$
which has $k$ common clause groups with $F(T_{n_{t}+1}, T_{i_{1}}, T_{i_{2}}, \cdots,  T_{i_{k}}, A)$ and
$T_{i_{k+1}}$ is a clause group from $B^{n_{t}}(x)$ (this clause group can be found because $n_{t} > 
k + 1$).

By the way we need prove that each variable of clause combination in unclearable value set of
relationship structure where each value set of clause combination consists of 1 value has the same value in
all clause combinations of that value of relationship structure. This result will also be used in next lemma.
That's easy to be shown.

Let $x_{i}$ - arbitrary variable presented in relationship structure. Let $F_{1} = F(T_{x_{i}x_{j}\cdots})$
and $F_{2} = F(T_{x_{i}x_{z}\cdots})$ - 2 different clause combinations which are parts of relationship
structure $R$. $V^{1}$ - unclearable values of relationship structure where each value set of clause
combination consists of 1 value. Let value $V_{1}^{F_{1}}$ of clause combination from $V_{1}$ which corresponds $F_{1}$ and
value $V_{1}^{F_{2}}$ of clause combination from $V_{1}$ which corresponds $F_{2}$ have different value of variable $x_{i}$.
Then operation C($V_{1}^{F_{1}}$, $V_{1}^{F_{2}}$) will give empty sets to both values. But that's contradiction because
values of relationship structure is unclearable.

The fact that $V_C$ is not empty and $V_{B}^{1} \subseteq V_{B}$ means that value of clause
combination $F(T_{i_{k+1}}, T_{i_{1}}, T_{i_{2}}, \cdots,  T_{i_{k}}, A)$ from $V_{B}^{1}$ is also a value
of the same clause combination from $V_{B}$ and from $V_{C}$. The fact that it can't be deleted during 
clearing means that exists value $V^{B}_{T_n}$ of clause combination $F(T_{n_{t}+1}, T_{i_{1}}, T_{i_{2}}, 
\cdots,  T_{i_{k}}, A)$ from $V_{C}$ which has the same values of common variables as value of 
$F(T_{i_{k+1}}, T_{i_{1}}, T_{i_{2}}, \cdots, T_{i_{k}}, A)$ from $V_{B}^{1}$. The only thing we need to 
prove now is that all clause combinations from $V_{C}$ which contain $T_{n_{t}+1}$ have value which can 
be added to $V_{B}^{1}$ and $V^{B}_{T_n}$ to create new value of relationship structure $V_{C}^{1}$  
which is unclearable.

Let's notice that these clause combinations don't give any new variables to clause combinations of $R_{B}$
and $F(T_{n_{t}+1}, T_{i_{1}}, T_{i_{2}}, \cdots,  T_{i_{k}}, A)$. This fact and the fact that in 
$V_{B}^{1}$ all values of the same variables in different clause combinations are the same can give us 
a hint that value of each clause combination which contains $T_{n_{t}+1}$ consisted of the same variable 
values as they presented in $V_{B}^{1}$ and value of clause combination $F(T_{n_{t}+1}, T_{i_{1}}, T_{i_{2}},
 \cdots,  T_{i_{k}}, A)$ discussed in previous paragraph.

$\Leftarrow$\\
This side is evident: the fact that $\exists$ $V_{1} \subseteq V_{res}$ means that $V_{res} \neq 
\varnothing$. \\
Lemma is proved.
\end{proof}

\begin{mylemma}
Let $V_{1}$ - value set of relationship structure based on k-CNF $A(x)$ where each value 
set of clause combination consists of 1 value of this clause combination. $V_{1}$ is unclearable
$\Leftrightarrow$ k-CNF $A(x)$ is equal to 1 on this value set.
\end{mylemma}
\begin{proof}
$\Rightarrow$\\
It was proved in Lemma 1 that corresponding variables have the same values in different clause combinations. Let's have a glance at k-CNF which variables values are the same as in the structure. It's evident that such 
k-CNF is equal to 1. Indeed for each clause exists clause combination that involves this clause. Clause 
combination is equal to 1 on this set $\Rightarrow$ clause itself is equal to 1. All clauses on this set 
are equal 1 $\Rightarrow$ k-CNF value on this set is equal 1.\\
 $\Leftarrow$\\
This proof is trivial. We take variable values $x_{12 \cdots m}$ that make k-CNF equal 1. It's evident 
that in value set of relationship structure $V_{1}$ based on $A(x)$ each value set of clause combination 
which is a member of $V_{1}$ and has the same variable values as $x_{12 \cdots m}$ is unclearable.\\
Lemma is proved.
\end{proof}

\begin{mytheorem}
Result of pair cleaning method applied to source k-CNF is not empty $\Leftrightarrow$ $\exists$ solution of
equation $k-CNF = 1$.
\end{mytheorem}
\begin{proof}
Consecutive usage of Lemma 1 and Lemma 2 proves the theorem.
\end{proof}

\begin{mytheorem}
Let \\
$V$ - value set of relationship structure based on k-CNF $A(x)$, \\
$V_{C} = C(V)$ - cleared value set of relationship structure, \\
$V_{C}^{1}$ - unclearable value set of relationship structure based on k-CNF $A(x)$ where each
value set of clause combination based on k-CNF consists of 1 value, \\
$V_{F_{i}}$ - value set of clause combination $F_{i}$, \\
$V_{F_{i}}^{C}$ - values of clause combination $F_{i}^{C}$, \\
$V_{F_{i}}^{0}$ - value of clause combination $F_{i}$. \\
Then
$V_{F_{i}}^{0} \in V_{F_{i}}$ - member of $V_{C} = C(V)$ $\Leftrightarrow$ $\exists V_{C}^{1}:
V_{F_{i}}^{0} \in V_{F_{i}}^{C}$ - member of $V_{C}^{1}$
\end{mytheorem}
\begin{proof}
Scheme of proof is the same as for Lemma 1, it's full description will be given a bit later.
\end{proof}

So we have not only algorithm for solving k-satisfiability problem but also algorithm for solving equation 
$A(x) = 1$. Of course in common case it's not polynomial (because number of solutions is $O(2^n)$). But 
process of getting each root of equation is polynomial. We'll describe it in full preprint version of this 
paper. 
\section{Complexity}
Number of values clause group can take is less than $2^{k}$.\\
Number of values clause combination can take is less than $2^{k(k+1)}$.\\
Number of clause combinations in relationship structure is $C_{n_{t}}^{k+1}$.\\
Number of comparisons during one iteration pass is less than $2^{2k(k+1)}(C_{n_{t}}^{k+1})^{2}$.\\
Number of iterations is less than $2^{k(k+1)}C_{n_{t}}^{k+1}$.\\
That means that number of operations for algorithm is less than $2^{3k(k+1)}(C_{n_{t}}^{k+1})^{3}$.\\
Therefore complexity of $k-SAT$ is $O(n_{t}^{3(k+1)})$.
For 3-SAT it's $O(n_{t}^{12})$.\\
$2^{-k}n \leq n_{t} \leq n$ $\Rightarrow$ method's complexity is $O(n^{3(k+1)})$.
For 3-SAT it's $O(n^{12})$. That means that pair cleaning method is polynomial and P=NP.

\end{document}